\documentclass[abstract]{scrartcl}

\usepackage[T1]{fontenc}
\usepackage[utf8]{inputenc}

\usepackage{amsmath}
\usepackage{amsfonts}
\usepackage{amssymb}
\usepackage{amsthm}
\usepackage{graphicx}
\usepackage[english]{babel}
\usepackage{enumerate}

\usepackage{mathtools}
\mathtoolsset{showonlyrefs}

\usepackage{bbm}

\usepackage{multicol}

\usepackage{xcolor}
\usepackage[breaklinks=true]{hyperref}
\hypersetup{ 
	colorlinks,
	linkcolor={blue!60!black},
	citecolor={blue!60!black},
	urlcolor={blue!60!black},
	linktoc=page}

\usepackage[numbers]{natbib}
\usepackage{authblk} % for /affil command

\usepackage{mathrsfs} % for \mathscr

\usepackage{soul} %st

\usepackage{todonotes}

\newcommand{\defn}{\textbf}

\renewcommand{\d}{\mathsf{d}}
\newcommand{\dd}{\mathbbm{d}}
\newcommand{\cL}{\mathcal{L}}

\newcommand{\var}[3]{\frac{\delta_{#1} {#2}}{\delta {#3}}}
\newcommand{\der}[2]{\frac{\partial {#1}}{\partial {#2}}}

\newcommand{\C}{\mathbb{C}}
\newcommand{\N}{\mathbb{N}}
\newcommand{\R}{\mathbb{R}}
\newcommand{\Z}{\mathbb{Z}}

\newcommand{\p}{\overline} %Plus shift
\newcommand{\m}{\underline} %Minus shift

\DeclareMathOperator{\D}{D}

\renewcommand{\t}{\mathcal T}

\newtheorem{thm}{Theorem}%[section]
\newtheorem{prop}[thm]{Proposition}
\newtheorem{lemma}[thm]{Lemma}
\newtheorem{cor}[thm]{Corollary}

\theoremstyle{definition}
\newtheorem{definition}[thm]{Definition}

\newtheorem{remark}[thm]{Remark}

\numberwithin{equation}{section}

\title{Semi-discrete Lagrangian 2-forms and the Toda hierarchy}

\author[1]{Duncan Sleigh}
\author[2]{Mats Vermeeren\thanks{Corresponding author. Email: \texttt{m.vermeeren@lboro.ac.uk}}}

\affil[1]{\normalsize School of Mathematics, University of Leeds}

\affil[2]{\normalsize Department of Mathematical Sciences, Loughborough University}

\date{}

\begin{document}
	
	\maketitle
	
	\begin{abstract}
		\noindent
		We present a variational theory of integrable differential-difference equations (semi-discrete integrable systems). This is an extension of the ideas known by the names ``Lagrangian multiforms'' and ``Pluri-Lagrangian systems'', which have previously been established in both the fully discrete and fully continuous cases. The main feature of these ideas is to capture a hierarchy of commuting equations in a single variational principle. Semi-discrete Lagrangian multiforms provide a new way to relate differential-difference equations and PDEs. We discuss this relation in the context of the Toda lattice, which is part of an integrable hierarchy of differential-difference equations, each of which involves a derivative with respect to a continuous variable and a number of lattice shifts. We use the theory of semi-discrete Lagrangian multiforms to derive PDEs in the continuous variables of the Toda hierarchy, which hold as a consequence of the differential-difference equations, but do not involve any lattice shifts. As a second example, we briefly discuss the semi-discrete potential KdV equation, which is related to the Volterra lattice.
		%\textbf{Keywords:}
		%
		%\medskip
		%\noindent
		%\textbf{MSC2010:} 
		%
	\end{abstract}

	\section{Introduction}
	
	Integrable systems can be described in many different ways, but some of the most important notions of integrability are formulated in the language of Hamiltonian dynamics. These ideas go back at least as far as Liouville, but a similar Lagrangian description of integrability is much more recent. It was first proposed in the context of integrable lattice equations, where all independent variables are discrete (see e.g.\@ \cite{boll2014integrability, lobb2009lagrangian, lobb2009kp}). Later it was developed in the fully continuous case as well, describing families of commuting ODEs or PDEs (see e.g.\@ \cite{suris2013variational,suris2016lagrangian,xenitidis2011Lagrangian-structure}).
	
	In the case of a family of ODEs, each equation is given its own independent variable. The euclidean space spanned by all these variables is called multi-time. The assumption that the ODEs commute means that solutions can be understood as functions of multi-time, rather than functions of a single time variable. In the case of PDEs, the members of an integrable family typically share their space variables, but again they are each given their own time variable. In this case, multi-time is spanned by both the common space-variables and the individual time variables. In the fully discrete case, multi-time is a lattice $\Z^N$ instead of a continuous euclidean space.
	
	The variational formulation of integrability has been presented in two subtly different ways, under the names ``Lagrangian multiforms'' and ``Pluri-Lagrangian systems''. It involves a differential form on multi-time. If we are dealing with ODEs, this is a 1-form. If we are considering a hierarchy of PDEs, it is a $d$-form, where $d$ is the number of independent variables of each individual member of the hierarchy. We can integrate this $d$-form over any orientable $d$-dimensional submanifold of multi-time. The variational principle requires that all such action integrals are critical. We can recover the usual action of one member of the hierarchy by taking the submanifold to be a coordinate (hyper)plane, but many other choices are possible. Hence the complete set of ``multi-time Euler-Lagrange equations'' is larger than the set set of Euler-Lagrange equations of the actions of each individual equation. The additional equations can be thought of as compatibility conditions between the coefficients of the Lagrangian $d$-form, in a similar sense to how vanishing Poisson brackets are a compatibility condition between the Hamiltonians of a Liouville integrable system. Hence a suitably chosen Lagrangian $d$-form can describe an integrable hierarchy in a consistent way.
	
	Connections have been established between the Lagrangian multiform approach and classical topics in integrable systems such as Hamiltonian structures \cite{suris2013variational,vermeeren2021hamiltonian}, variational symmetries \cite{petrera2017variational,petrera2021variational,sleigh2020variational}, and Lax pairs \cite{sleigh2019variational}.
	
	In the present work we extend the theory of Lagrangian multiforms to the semi-discrete case, where some of the independent variables are continuous but others discrete. An application to semi-discrete systems was proposed in one of the early works on Lagrangian multiforms \cite{yoo2011discrete}, but a systematic development of this case has not been carried out before. Our main example will be the hierarchy consisting of the Toda lattice and its symmetries, which together form the Toda hierarchy. The first two equations of this hierarchy are
	\begin{align*}
		&q_{11} = \exp(\p q - q) - \exp( q - \m q), \\
		&q_2 = q_1^2 + \exp(\p q - q) + \exp( q - \m q),
	\end{align*}
where subscripts denote derivatives with respect to the continuous independent variables $t_1$ and $t_2$, and the bar and underline denote lattice shifts in opposite directions. When considering only these two equations, we can take $q$ to be a function of $\Z \times \R^2$, so our multi-time is described by one discrete and two continuous variables. When considering additional members of the hierarchy, with time variables $t_3,\ldots,t_N$, our multi-time will be $\Z \times \R^N$.
	
	A continuous Lagrangian 1-form for the Toda hierarchy was given in \cite{petrera2017variational}. In that description, the elements of the configuration space are vectors describing the positions of all particles. In particular, the discrete direction is not treated as an independent variable, so the multi-time in this case is $\R^N$. In this work we present a semi-discrete Lagrangian 2-form for the Toda hierarchy, in which the lattice position is a discrete independent variable, i.e.\@ multi-time is $\Z \times \R^N$. The semi-discrete components of this 2-form are closely related to the aforementioned 1-form, but the doubly continuous components are a new feature. From these doubly continuous components we will derive PDEs that hold on each single lattice site. These PDEs do not involve any lattice shifts, but hold as a consequence of the lattice equations of the Toda hierarchy.
	
The task at hand is to introduce Lagrangian 2-form theory in the setting of a semi-discrete multi-time $\Z \times \R^N$. This will require us to define the notions of a semi-discrete differential form and of a semi-discrete submanifold of $\Z \times \R^N$. The central principle of semi-discrete Lagrangian multiform theory can then be formulated in terms of the action integrals obtained by integrating a semi-discrete 2-form over an arbitrary semi-discrete surface within $\Z \times \R^N$. Although in the example of the Toda hierarchy the discrete direction has the interpretation of space and the continuous variables can be thought of as times, this interpretation plays no role on the general theory.

	The plan for the paper is as follows. In Section \ref{sec:geometry} we introduce the notions of semi-discrete manifolds and semi-discrete differential forms. In Section \ref{sec:theory} the theory of semi-discrete Lagrangian multiforms is developed.
	In Section \ref{sec:Toda} we derive a semi-discrete Lagrangian 2-form for the Toda lattice and study its implications.  In Section \ref{sec:sdpkdv} we briefly present a second example: the semi-discrete potential KdV hierarchy, which is closely related to the Volterra lattice. We close the paper with a few concluding thoughts and an appendix containing the computations required to generalise the theory to higher semi-discrete forms.
	
	\section{Semi-discrete geometry}
	\label{sec:geometry}
	
	In this section we present the necessary concepts of semi-discrete geometry. For ease of presentation we assume throughout the main text that there is only one discrete dimension. However, all concepts can be extended to a context with several discrete directions, as is discussed in the appendix.
	
	A semi-discrete surface in $\Z \times \R^N$ is a collection of surfaces and curves in $\R^N$, which are each assigned a value of $\Z$.  A possible intuition is that the curves represent the locations where the surface ``jumps'' to a different value of $\Z$. An example is shown in Figure \ref{fig-sd}. This intuition is limited, however, because the semi-discrete surfaces that have the most obvious dynamical meaning consist entirely of lines. If $\Z \times \R^N$ is the space of independent variables $(k, t_1, \ldots, t_N)$ of the Toda hierarchy, then the $n$-th Toda equation can be considered on the subspace with $t_1, \ldots, t_{n-1}, t_{n+1}, \ldots, t_N$ fixed, which is a semi-discrete surface consisting of a line in the $t_n$-direction at each lattice site. Furthermore, this is the semi-discrete surface which we would integrate over in the variational principle for the $n$-th Toda equation by itself. To obtain a variational description of the hierarchy as a whole, i.e.\@ a semi-discrete Lagrangian 2-form, more general semi-discrete surfaces will be needed.
	
	To be precise, we have the following definition:
	
	\begin{figure}[t]
		\centering
		\includegraphics[width=.75\linewidth]{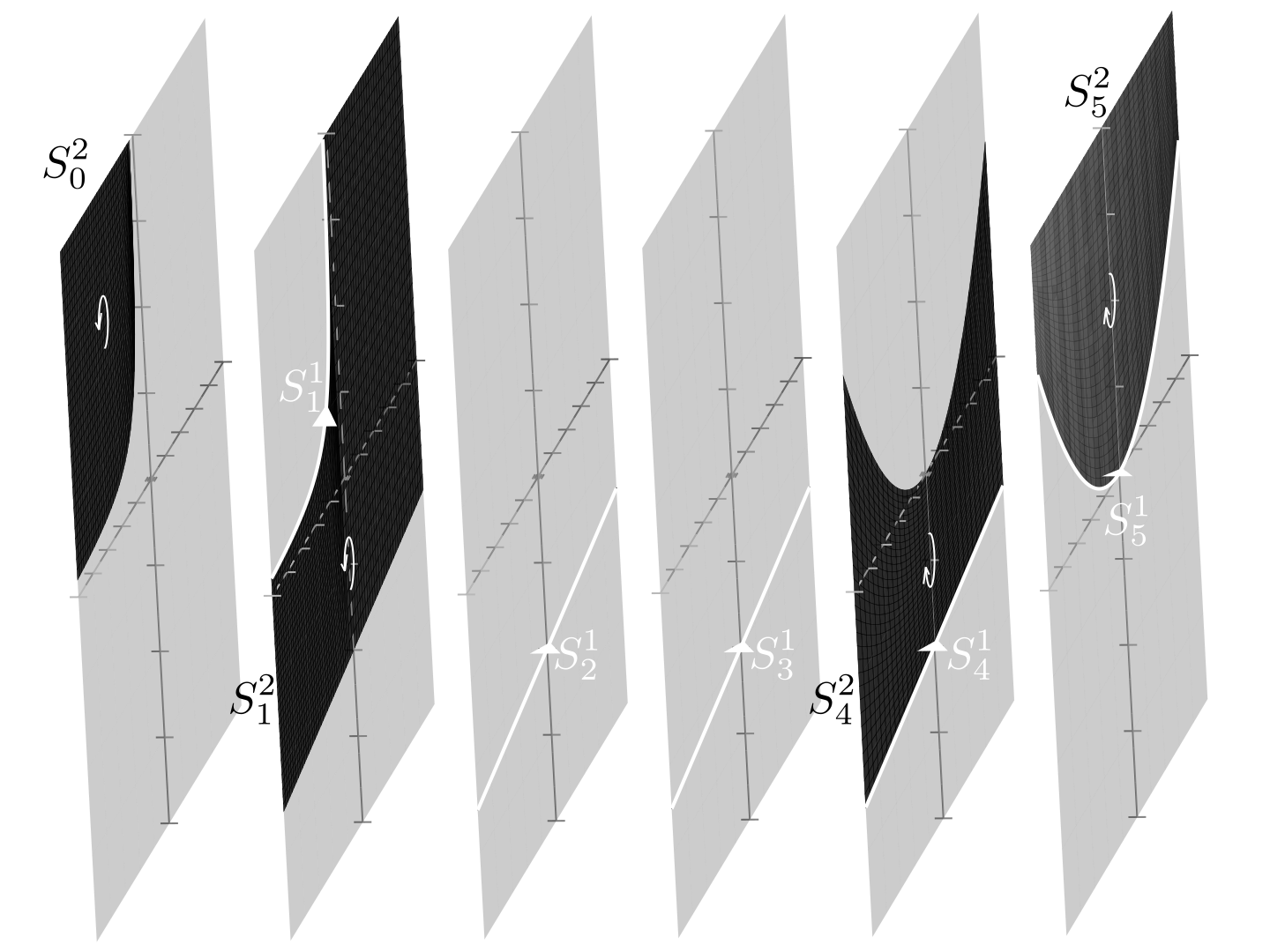}%
		\caption{Visualisation of the space $\Z \times \R^2$ and a semi-discrete surface inside of it. The 2-dimensional elements of the semi-discrete surface are shown in black and the 1-dimensional elements in white. Note that we have oriented the 1-dimensional elements opposite to the boundaries of the 2-dimensional elements, so that these cancel each other when taking the boundary of the semi-discrete surface.}
		\label{fig-sd}
	\end{figure}
	
	\begin{definition}
		\begin{enumerate}[(a)]
			\item
			A $d$-dimensional \defn{semi-discrete submanifold} $S$ of $\Z \times \R^N$ is pair of disjoint unions
			\[
			S = \left( \bigsqcup_{k \in \Z} S_k^{d-1} , \bigsqcup_{k \in \Z} S_k^{d} \right)
			\]
			where $S_k^{d-1}$ is a disjoint union of oriented $(d-1)$-dimensional submanifolds of $\R^N$ and $S_k^d$ is a disjoint union of oriented $d$-dimensional submanifolds of $\R^N$. 
			
			The disjoint unions $S_k^{d-1}$ and $S_k^d$ represent the continuous and semi-discrete elements of the semi-discrete surface at discrete position $k \in \Z$. They are defined as disjoint unions of submanifolds, rather than submanifolds themselves, to allow for overlapping elements and count their multiplicity. 
			
			\item If $d = 2$ we speak of a \defn{semi-discrete surface} and if $d=3$ of a \defn{semi-discrete volume}.
			
			\item The \defn{boundary} $\partial S$ of the $d$-dimensional semi-discrete submanifold $S$ is a $(d-1)$-dimensional semi-discrete submanifold given by
			\[ \partial S = \left( \bigsqcup_{k \in \Z} -\partial S_k^{d-1} , \bigsqcup_{k \in \Z} \left( \partial S_k^{d} \sqcup S_k^{d-1} \sqcup - S_{k+1}^{d-1} \right) \right) , \]
			where the minus sign denotes a change of orientation.
			
			Note that the sign conventions are chosen to ensure that the boundary of a boundary is empty (modulo disjoint unions of two copies of the same manifold with opposite orientations).
		\end{enumerate}    
	\end{definition}

	\begin{remark}
		To make sure a semi-discrete surface looks like a discretisation of a smooth surface, we could require that the $S_k^{d-1}$ and $S_k^d$ do not contain overlapping elements and that the corresponding subset 
		\[ \left( \bigcup_{k} \{k\} \times S_k^d  \right) \cup \left( \bigcup_k [k, k + 1] \times S_k^{d-1} \right) \subset \R \times \R^N \]
		is a topological manifold. However, this restriction is not needed in any of the following. A related, but more useful, restriction would be to consider only semi-discrete submanifolds without boundary. Just like in the classical calculus of variations, we will impose that variations vanish on the boundary of the submanifold on which the action is defined. Considering submanifolds without boundary would remove this condition.
		Examples of semi-discrete submanifolds without boundary include:
		\begin{itemize}
		\item An unbounded continuous manifold at one lattice site: $S_k^{d-1} = \emptyset$, $S_k^{d} = \emptyset$ for $k \neq k_0$, and $S_{k_0}^{d}$ a $d$-dimensional manifold without boundary.
		\item Copies of the same $(d-1)$-dimensional manifold in all lattice locations: $S_k^{d-1}$ independent of $k$ and $S_k^{d} = \emptyset$.
		\item The semi-discrete surface shown in Figure \ref{fig-sd}, assuming it is continued in a suitable way beyond the edges of the image.
		\end{itemize}
		At the other extreme, a simple example of a semi-discrete submanifold in $\Z \times \R^2$ with a very large boundary can be constructed by intersecting the inclined plane $\{(t_0,t_1,t_0) \mid t_0,t_1 \in \R\} \subset \R^3$ with $\Z \times \R^2$. We get:
		\begin{itemize}
		\item $S_k^{d-1} = \{(t_1,k) \mid t_1 \in \R\}$ and $S_k^{d} = \emptyset$. It consists only of lines, and it is contained within its boundary, which is given by
			\[ \partial S = \left( \emptyset , \bigsqcup_{k \in \Z} \left( S_k^{d-1} \cup - S_{k+1}^{d-1} \right) \right) . \]
			This surface will not be of interest in the variational principle, since any variation that vanishes on the boundary vanishes on the whole surface.
		\end{itemize}
	\end{remark}
	
	The semi-discrete space $\Z \times \R^N$ will be our space of independent variables, which we call multi-time.
	We consider semi-discrete fields $q: \Z \times \R^N \mapsto Q$, taking values in some configuration space $Q$. Often we will have $Q = \R$ or $Q = \C$.
	When there is no risk of confusion we will write $q^{[n]}$ or simply $q$ for $q(n,t_1,\ldots,t_N)$.
	To denote partial derivatives of $q$ we will use a multi-index notation. A multi-index is an $N$-tuple $I = (i_1,\ldots,i_N)$ of non-negative integers. We define
	\[ q_I = \der{^{i_1}}{t_1^{i_1}} \ldots \der{^{i_N}}{t_N^{i_N}} q, \]
	so that each entry of $I$ states the number of derivatives to be taken with respect to the corresponding time variable. We will use the notations $I t_j$ and $I \setminus t_j$ to raise or lower an entry of $I$, i.e.\@
	\begin{align*}
	&I t_j = (i_1, \ldots, i_j + 1, \ldots, i_N) , \\
	&I \setminus t_j = (i_1, \ldots, i_j - 1, \ldots, i_N) \qquad \text{if } i_j > 0.
	\end{align*}
	We write $I \not\ni t_j$ if $i_j = 0$.
	By $\t$ we denote the shift operator: $\t q^{[n]} = q^{[n+1]}$, i.e.\@
	\[ \t q(n,t_1,\ldots,t_N) = q(n+1,t_1,\ldots,t_N). \]
	
	We denote by $\mathcal{Q}$ the set of all semi-discrete fields. 
	We are interested in functions of the semi-discrete fields that are autonomous (only depend on $\Z \times \R^N$ through $q \in \mathcal{Q}$) and local in the sense that $f\!\left[q^{[n]}\right]$ depends on $q^{[n+k]} = q(n+k,t_1,\ldots,t_N)$ and its derivatives for a finite number of $k \in \Z$. Since $f$ is assumed to be autonomous, the shift operator acts on it as $\t f[q] = f[\t q]$.
	
	\begin{definition}\label{def-form}
		\begin{enumerate}[(a)]
			\item 
			A $\mathcal Q$-dependent \defn{semi-discrete $d$-form} on $\Z \times \R^N$ is a pair
			\[ \cL[q] = \left(\cL^{d-1}[q] \,,\, \cL^d[q] \right) \]
			consisting of a $(d-1)$-form and a $d$-form, with coefficients that are functions of $\mathcal Q$ in the sense explained above. 
			
			\item
			The \defn{semi-discrete integral} of $\cL$ over a $d$-dimensional semi-discrete submanifold $S$ is given by
			\[ \int_S \cL[q] = \sum_k \int_{S^{d-1}_k} \cL^{d-1}\!\left[q^{[k]}\right] + \sum_k \int_{S^d_k} \cL^d\!\left[q^{[k]}\right], \]
			where the integral over a disjoint union of submanifolds is understood as the sum of the integrals over each of the submanifolds.
			
			\item
			The \defn{exterior derivative} of $\cL[q] = \left(\cL^{d-1}[q] , \cL^d[q] \right)$ is a $\mathcal Q$-dependent semi-discrete $(d+1)$-form defined by
			\begin{equation}\label{exterior-der}
			\dd \cL = \left(\Delta (\cL^d) - \d \cL^{d-1} \,,\, \d \cL^d \right) ,
			\end{equation}
			where $\Delta = \mathsf{id} - \t^{-1}$ is the backward difference operator.
		\end{enumerate}
	\end{definition}
	
	The signs in Equation \eqref{exterior-der} are chosen such that we get the usual alternating expressions in terms of the coefficients of $\cL$. Indeed, if we write 
	\[ \cL = \left( \sum_{i_1<\ldots<i_{d-1}} L_{0 i_1 \ldots i_{d-1}} \,\d t_{i_1} \wedge \ldots \wedge \d t_{i_{d-1}} \,, \sum_{i_1<\ldots<i_{d}} L_{i_1 \ldots i_{d}} \,\d t_{i_1} \wedge \ldots \wedge \d t_{i_{d}} \right) \]
	then we find 
	\begin{align*} \dd \cL 
	&= \left( \sum_{i_1<\ldots<i_{d}} \left( \Delta L_{i_1 \ldots i_{d}} + \sum_{\alpha=1}^d (-1)^{\alpha} \D_{i_\alpha} L_{0 i_1 \ldots \widehat{i_\alpha} \ldots i_d} \right) \d t_{i_1} \wedge \ldots \wedge \d t_{i_{d}} \,, \right. \\
	&\qquad \left. \sum_{i_1<\ldots<i_{d+1}} \sum_{\alpha=1}^{d+1} (-1)^{\alpha-1} \D_{i_{\alpha}} L_{i_1 \ldots \widehat{i_\alpha} \ldots i_{d+1}} \,\d t_{i_1} \wedge \ldots \wedge \d t_{i_{d+1}} \right) ,
	\end{align*}
	where
	\[ \D_i=\frac{\partial}{\partial t_i} +\sum_I q_{It_i}\frac{\partial}{\partial q_I} \]
	is the total derivative with respect to $t_i$.
	
	Our definition of the exterior derivative is justified by the following semi-discrete version of Stokes' theorem.
	
	\begin{thm}
		\label{thm:stokes}
		Let $S$ be a $(d+1)$-dimensional semi-discrete submanifold of $\Z \times \R^N$ and $\cL$ a $\mathcal Q$-dependent semi-discrete $d$-form. There holds
		\[ \int_S \dd \cL = \int_{\partial S} \cL. \]
	\end{thm}
	\begin{proof}
		Using our definitions and the smooth Stokes theorem we find
		\begin{align*}
		\int_{S} \dd \cL
		&= \sum_k \int_{S^{d}_k} \left( \Delta \cL^d\!\left[q^{[k]}\right] - \d \cL^{d-1}\!\left[q^{[k]}\right] \right) + \sum_k \int_{S^{d+1}_k} \d \cL^d\!\left[q^{[k]}\right] \\
		&= \sum_k \left( \int_{S^{d}_k} \cL^d\!\left[q^{[k]}\right] - \int_{S^{d}_k} \cL^d\!\left[q^{[k-1]}\right] - \int_{\partial S^{d}_k} \cL^{d-1}\!\left[q^{[k]}\right] \right) + \sum_k \int_{\partial S^{d+1}_k} \cL^d\!\left[q^{[k]}\right] \\
		&= -\sum_k \int_{\partial S^{d}_k} \cL^{d-1}\!\left[q^{[k]}\right]
		+ \sum_k \left( \int_{\partial S^{d+1}_k} \cL^d\!\left[q^{[k]}\right] + \int_{S^{d}_k} \cL^d\!\left[q^{[k]}\right]  - \int_{S^{d}_{k+1}} \cL^d\!\left[q^{[k]}\right]  \right)\!\!\!\!\!\! \\
		&= \int_{\partial S} \cL. \qedhere
		\end{align*}
	\end{proof}
	
	We close this section with an important lemma about commuting the shift operator and the total derivative with partial derivatives.
	\begin{lemma}\label{lemma-commute-ders}
		There holds
		\begin{align}
		& \der{\t^k f}{q_I} = \t^k \der{f}{\t^{-k} q_{I}} , \label{partialT} \\
		& \der{(\D_j f)}{q_I} = \der{f}{q_{I \setminus t_j}} + \D_j \der{f}{q_I} , \label{partialD}
		\end{align}
		where ${I \setminus t_j}$ is the multi-index obtained from $I$ be reducing the $j$-th entry by one, and the term containing it is taken to be zero if $I \not\ni t_j$.
	\end{lemma}
	\begin{proof}
		Equation \eqref{partialT} is a direct consequence of the fact that $\t$ acts on all instances of $q$ in the expression to the right of it.
		To derive Equation \eqref{partialD} we calculate
		\begin{align*}
		\der{(\D_j f)}{q_I}
		&= \der{}{q_I} \left( \sum_J \der{f}{q_J} q_{J t_j}\right) 
		= \sum_J \der{f}{q_I \partial q_J} q_{J t_j} + \der{f}{q_{I \setminus t_j}} 
		= \D_j \der{f}{q_I} + \der{f}{q_{I \setminus t_j}}. \qedhere
		\end{align*}
	\end{proof}
	
	\section{Semi-discrete Lagrangian multiforms}
	\label{sec:theory}
	
	In the continuous Lagrangian multiform theory, the central object is a Lagrangian $d$-from, which is integrated over arbitrary $d$-dimensional submanifolds of multi-time. In the following, this role will be played by a $\mathcal Q$-dependent semi-discrete 2-form $\cL[q] = \left( \cL^1[q] , \cL^2[q] \right)$ with
	\[ \cL^1[q] = \sum_{j>0} L_{0j} [q] \, \d t_j \]
	and
	\[ \cL^2[q] = \sum_{j > i > 0} L_{ij} [q]\,\d t_i \wedge \d t_j . \]
	If $j < i$ we define $L_{ij} = -L_{ji}$.
	
	\begin{definition}
		\label{def-critical}
		We say that a semi-discrete field $q:\Z \times \R^N \to Q$ is \defn{critical} for $\cL[q]$ if for every semi-discrete surface $S$ it satisfies
		\begin{equation}
		\label{critical}
		\der{}{\varepsilon}\bigg|_{\varepsilon = 0} \int_S \cL[q + \varepsilon v] = 0 
		\end{equation}
		for any field $v$ that vanishes (along with all of its derivatives) at the boundary of $S$.
	\end{definition}
	
	Equation \eqref{critical} can also be written as
	\[ \delta \int_S \cL[q] = \int_S \delta \cL[q] = 0 \]
	where $\delta$ is the Gateaux derivative in a direction to be specified (the arbitrary $v$ in Definition \ref{def-critical}). This $\delta$ can also be understood as the vertical exterior derivative in the variational bicomplex  (see for example \cite{anderson1992introduction} or \cite[Appendix A]{suris2016lagrangian}).
	
	\begin{definition}
		\label{def-varder}
		The 1- and 2-dimensional continuous \defn{variational derivatives} of a function $P$, with respect to $q_I$, are defined as
		\begin{align*}
		&\var{i}{P}{q_I} = \sum_{\alpha \in \N} (-1)^{\alpha} \D_i^\alpha \der{P}{q_{I t_i^\alpha}} , \\
		&\var{ij}{P}{q_I} = \sum_{\alpha,\beta \in \N} (-1)^{\alpha+\beta} \D_i^\alpha \D_j^\beta \der{P}{q_{I t_i^\alpha t_j^\beta}} ,
		\end{align*}
		where $\N = \{0,1,2,\ldots\}$.
		The \defn{semi-discrete variational derivatives} of a function $P$, with respect to $q_I$, are defined as
		\noeqref{sd-varder1,sd-varder2,sd-varder3}%to avoid labels being supressed
		\begin{subequations}
			\label{sd-varders}
			\begin{align}
			&\var{0}{P}{q_I} = \der{}{q_{I}} \sum_{n \in \N} \t^{-n} P , \label{sd-varder1}\\
			&\var{0i}{P}{q_I} = \var{i}{}{q_{I}} \sum_{n \in \N} \t^{-n} P , \label{sd-varder2}\\
			&\var{0ij}{P}{q_I} = \var{ij}{}{q_{I}} \sum_{n \in \N} \t^{-n} P. \label{sd-varder3}
			\end{align}
		\end{subequations}
	\end{definition}
	
	To give a few examples, denoting $\bar q = \t q$ and $\underline q = \t^{-1} q$, we have
	\begin{align*}
	\var{0i}{q_{t_i}^2}{q} &= - 2\D_i q_{t_i} = -2q_{t_it_i}, \\
	\var{0ij}{\bar q_{t_i}^2}{q} &= - 2\D_i q_{t_i} = -2q_{t_it_i}, \\
	\var{0i}{q \bar q}{q} &= \bar q + \underline q, \\
	\var{0i}{q \underline q}{q} &= \underline q.
	\end{align*}
	The last example highlights the fact that there are no positive shifts in the right hand sides of Equations \eqref{sd-varders}.
	
	\begin{remark}
		The familiar variational derivative $\var{i}{}{q}$ is part of an exact complex, satisfying $\var{i}{}{q} \circ \D_i = 0$. This property fails for some of the variational derivatives of Definition \ref{def-varder}. For example, we have $\big(\var{i}{}{q_i} \circ \D_i \big) q = 1 \neq 0$. The analogous discrete property also fails, for example $\big(\var{0}{}{q} \circ \Delta \big) q = 1 \neq 0$. We still use the term ``variational derivative'' because these are the expressions that we encounter in the calculus of variations in multi-time.
	\end{remark}
	
	\begin{prop}\label{prop-mEL-P}
		The following are equivalent:
		\begin{enumerate}[(i)]
			\item\label{prop-mEL-deltadL-1} The field $q$ is critical.
			\item\label{prop-mEL-deltadL-2} $\delta \dd \cL = 0$.
			\item\label{prop-mEL-deltadL-3} For all multi-indices $I$ and all $n$ there holds
			\begin{equation}
			\label{var-P}
			\var{0ij}{P_{0ij}}{q_I^{[n]}} = 0 \qquad\text{and}\qquad \var{ijk}{P_{ijk}}{q_I^{[n]}} = 0 .
			\end{equation}
			
		\end{enumerate}
	\end{prop}
	\begin{proof}
		Assume that $q$ is critical. Consider an arbitrary semi-discrete volume $V$ and integrate $\cL$ over its boundary. By Theorem \ref{thm:stokes} we have
		\[ \int_{\partial V} \cL =  \int_V \dd \cL .\]
		Since $q$ is critical, infinitesimal variations of the left hand side must vanish. It follows that
		\[ \int_V \delta \dd \cL = \delta \int_V \dd \cL = 0. \]
		Since $V$ is an arbitrary volume, it follows that $\delta \dd \cL = 0$.
		
		Following the above steps in reverse, we can see that if $\delta \dd \cL = 0$, then the variational principle is satisfied on all semi-discrete surfaces that are the boundary of some semi-discrete volume. One can show that this implies that the variational principle is satisfied on all semi-discrete surfaces. To do this, it is sufficient to observe that the variational principle can be restricted without loss of generality to variations with an arbitrarily small support, and that every discrete surface is \emph{locally} the boundary of some semi-discrete volume.
		
		To prove that \eqref{prop-mEL-deltadL-2} and \eqref{prop-mEL-deltadL-3} are equivalent, we will show that Equation \eqref{var-P} holds for all multi-indices $I$ if and only if 
		\begin{equation}
		\label{der-P}
		\der{P_{0ij}}{q_I^{[n]}} = 0 \qquad\text{and}\qquad \der{P_{ijk}}{q_I^{[n]}} = 0
		\end{equation}
		for all multi-indices $I$ and all $n$. This establishes the claimed equivalence, because the left hand sides of Equation \eqref{der-P} are the coefficients of $\delta \dd \cL$. The implication from Equation \eqref{der-P} to Equation \eqref{var-P} follows immediately from the definition of the variational derivatives. To prove the opposite implication, observe that we can write a partial derivative in terms of variational derivatives:
		\[
		\der{P_{0ij}}{q_I^{[n]}} =  \sum_{\alpha,\beta \in \{0,1\}} \left( \D_i^\alpha \D_j^\beta \var{0ij}{P_{0ij}}{q_{I t_i^\alpha t_j^\beta}^{[n]}} - \t^{-1} \D_i^\alpha \D_j^\beta \var{0ij}{P_{0ij}}{q_{I t_i^\alpha t_j^\beta}^{[n+1]}} \right)
		\]
		and
		\[
		\der{P_{ijk}}{q_I^{[n]}} = \sum_{\alpha,\beta,\gamma \in \{0,1\}} \D_i^\alpha \D_j^\beta \D_k^\gamma \var{ijk}{P_{ijk}}{q_{I t_i^\alpha t_j^\beta t_k^\gamma}^{[n]}}. \qedhere
		\]
	\end{proof}
	
	Property $(ii)$ of Proposition \ref{prop-mEL-P} will be useful later on, because it is satisfied if the coefficients of $\dd \cL$ are products of two factors that vanish on the equations of motion (or are sums of such products). Hence, if we construct a semi-discrete 2-form such that $\dd \cL$ attains such a ``double zero'' on solutions to a set of equations, then it is guaranteed that this set of equations implies the multi-time Euler-Lagrange equations. The equivalence between $(i)$ and $(iii)$ will be used in the proof of the following theorem.
	
	\begin{subequations}
		\begin{thm}
			\label{thm-EL}
			A field is critical if and only if all of the following \defn{multi-time Euler-Lagrange equations} hold for all $n \in \Z$:
			\begin{align}
			&\var{ij}{L_{ij}}{q_{I}^{[n]}} = 0 & \forall I \not\ni t_i,t_j , \label{EL-1}
			\\
			&\var{ij}{L_{ij}}{q_{I t_j}^{[n]}} - \var{ik}{L_{ik}}{q_{I t_k}^{[n]}} = 0 & \forall I \not\ni t_i , \label{EL-2}
			\\
			&\var{ij}{L_{ij}}{q_{I t_i t_j}^{[n]}} + \var{jk}{L_{jk}}{q_{I t_j t_k}^{[n]}} + \var{ki}{L_{ki}}{q_{I t_k t_i}^{[n]}}= 0 & \forall I , \label{EL-3}
			\\
			&\var{ij}{L_{ij}}{q_{I t_j}^{[n]}} + \var{0i}{L_{0i}}{q_{I}^{[n]}} = 0 & \forall I \not\ni t_i , \label{EL-D1}\\
			&\var{ij}{L_{ij}}{q_{I t_i t_j}^{[n]}} - \var{0j}{L_{0j}}{q_{I t_j}^{[n]}} + \var{0i}{L_{0i}}{q_{I t_i}^{[n]}} = 0 & \forall I . \label{EL-D2}
			\end{align}
		\end{thm}
		If $n$ is such that $L_{ij}$ does not depend on $q_I^{[n]}$ for any $I$, then it follows from \eqref{EL-D1} and \eqref{EL-D2} that 
		\noeqref{EL-D3,EL-D4} % to ensure labels don't get supressed
		\begin{align}
		& \var{0i}{L_{0i}}{q_{I}^{[n]}} = 0 & 
		\forall I \not\ni t_i , \label{EL-D3} 
		\\
		&\var{0j}{L_{0j}}{q_{I t_j}^{[n]}} - \var{0i}{L_{0i}}{q_{I t_i}^{[n]}} = 0 & \forall I . \label{EL-D4}
		\end{align}
	\end{subequations}%
	\begin{proof}[Proof of Theorem \ref{thm-EL}]
		We write $q$ for $q^{[n]}$, hence $\t^m q = q^{[m+n]}$. Using Lemma \ref{lemma-commute-ders} we find, for any multi-index $J$,
		\begin{align*}
		\der{P_{0ij}}{q_J} &= \der{L_{ij}}{q_{J}} - \t^{-1} \der{L_{ij}}{\t q_J} - \D_i \der{L_{0j}}{q_J} - \der{L_{0j}}{q_{J \setminus t_i}} + \D_j \der{L_{0i}}{q_J} + \der{L_{0i}}{q_{J \setminus t_j}} .
		\end{align*}
		Hence 
		\begin{align}
		\var{0ij}{P_{0ij}}{q_J} &= \sum_{m,\alpha,\beta \in \N} (-1)^{\alpha+\beta} \D_i^\alpha \D_j^\beta \der{\t^{-m} P_{0ij}}{q_{J t_i^\alpha t_j^\beta }} 
		\notag \\
		&= \sum_{m,\alpha,\beta \in \N} (-1)^{\alpha+\beta} \D_i^\alpha \D_j^\beta \left( \t^{-m} \der{L_{ij}}{\t^{m} q_{J t_i^\alpha t_j^\beta }} - \t^{-m-1} \der{L_{ij}}{\t^{m+1} q_{J t_i^\alpha t_j^\beta }} \right) \notag \\
		&\qquad - \sum_{m,\alpha,\beta \in \N} (-1)^{\alpha+\beta} \D_j^\beta \left(  \D_i^{\alpha+1} \der{\t^{-m} L_{0j}}{ q_{J t_i^\alpha t_j^\beta } } + \D_i^\alpha \der{\t^{-m} L_{0j}}{ q_{J t_i^{\alpha-1} t_j^\beta } } \right) \notag \\
		&\qquad + \sum_{m,\alpha,\beta \in \N} (-1)^{\alpha+\beta} \D_i^\alpha \left( \D_j^{\beta+1} \der{\t^{-m} L_{0i}}{ q_{J t_i^\alpha t_j^\beta } } + \D_j^\beta \der{\t^{-m} L_{0i}}{ q_{J t_i^\alpha t_j^{\beta-1} } } \right) 
		\notag \\
		&= \sum_{\alpha,\beta \in \N} (-1)^{\alpha+\beta} \D_i^\alpha \D_j^\beta \der{L_{ij}}{q_{J t_i^\alpha t_j^\beta }} \notag \\
		&\qquad - \sum_{m,\beta \in \N} (-1)^{\beta} \D_j^\beta \der{\t^{-m} L_{0j}}{ q_{J  t_i^{-1} t_j^\beta } }
		+ \sum_{m,\alpha \in \N} (-1)^{\alpha} \D_i^\alpha \der{\t^{-m} L_{0i}}{ q_{J t_i^\alpha t_j^{-1} } } \notag \\
		&= \var{ij}{L_{ij}}{q_{J}} - \var{0j}{L_{0j}}{q_{J t_i^{-1}}} + \var{0i}{L_{0i}}{q_{J t_j^{-1}}},
		\label{varder-P0ij}
		\end{align}
		where $J t_i^{-1} = J \setminus t_i$ denotes the multi-index obtained for $J$ by reducing the $i$-th entry by one if it is positive, and any term containing $J t_i^{-1}$ is taken to be zero if the $i$-th entry of $J$ is zero. Similarly, we find
		\begin{equation}
		\var{ijk}{P_{ijk}}{q_J} = \var{ij}{L_{ij}}{q_{J t_k^{-1}}}  + \var{jk}{L_{jk}}{q_{J t_i^{-1}}} + \var{ki}{L_{ki}}{q_{J t_j^{-1}}} .
		\label{varder-Pijk}
		\end{equation}
		
		By Proposition \ref{prop-mEL-P} it follows that the semi-discrete field is critical if and only if the expressions \eqref{varder-P0ij}--\eqref{varder-Pijk} equal zero for all $J$. Considering these equations for different types of multi-indices, we find equations \eqref{EL-1}--\eqref{EL-D2}:
		\begin{itemize}
			\item From Equation \eqref{varder-P0ij}, with $J = I \not\ni t_i,t_j$, we obtain Equation \eqref{EL-1}. If $J = I t_j$ with $I \not\ni t_i$ we find Equation \eqref{EL-D1}, and if $J = I t_i t_j$ we find Equation \eqref{EL-D2}
			
			\item From Equation \eqref{varder-Pijk} with $J = I t_k$, $I \not\ni t_i,t_j$, we obtain once again Equation \eqref{EL-1}. If $J = I t_j t_k$ with $I \not\ni t_i$ we find Equation \eqref{EL-2}, and if $J = I t_i t_j t_k$ we find Equation \eqref{EL-3}. If $J \not\ni t_i,t_j,t_k$, then \eqref{varder-Pijk} vanishes identically.
			\qedhere
		\end{itemize}
	\end{proof}
	
	\begin{remark}
	\label{remark-reduction}
	A semi-discrete Lagrangian 2-form $\cL$ on $\Z \times \R^N$ can be reduced to a continuous Lagrangian 1-form $\mathcal{M}$ on $\R^N$ by summing over the lattice sites. Let us illustrate this for the case of a periodic lattice, $\t^n q = q$. A continuous 1-form is obtained from the 1-form part of the semi-discrete 2-form $\cL = (\cL^1, \cL^2)$:
	\[ \mathcal{M} = \sum_{\alpha = 0}^{n-1} \cL^1 [q^{[\alpha]}] . \]
	Fix a curve $\gamma$ in $\R^N$ and consider the semi-discrete surface of integration $S = (S_k^1,S_k^2)$ with $S_k^1 = \gamma$ and $S_k^2 = \emptyset$. The semi-discrete integral of $\cL$ over $S$ is
	\[ \sum_{\alpha = 0}^{n-1} \int_\gamma L[q^{[\alpha]}] = \int_\gamma \mathcal{M} . \]
	Hence every action integral of the continuous 1-form $\mathcal{M}$ is also an action for the semi-discrete 2-form $\cL$, so every solution to the variational problem for the 2-form $\cL$ is also a solution to the variational problem for the 1-form $\mathcal{M}$.
	\end{remark}

	\section{Toda lattice}
	\label{sec:Toda}
	
	The Toda lattice \cite{toda1967vibration} is an integrable model consisting of $k$ particles on a line, with nearest-neighbour interaction. The deviation from equilibrium of one of the particles is given  by $q = q^{[n]} = q(n,t_1,t_2,\ldots)$. We use a bar-notation for shifts: 
	\[\p q = \t q = q(n+1,t_1,t_2,\ldots), \qquad \m q = \t^{-1} q = q(n-1,t_1,t_2,\ldots). \]
	For derivatives of $q$ we use the subscript notations: 
	\[ q_i = q_{t_i} = \der{q}{t_i}, \qquad q_{i j} = q_{t_i t_j} = \der{^2 q}{t_i \partial t_j}. \]
	
	The Toda lattice and the next two members of its hierarchy are given by 
	\noeqref{Toda-3} % To make sure label is not supressed
	\begin{subequations}
		\label{Toda}
		\begin{align}
		&q_{11} = \exp(\p q - q) - \exp( q - \m q) \label{Toda-1}, \\
		&q_2 = q_1^2 + \exp(\p q - q) + \exp( q - \m q) \label{Toda-2}, \\
		&q_3 = q_1^3 + (2 q_1 + \m q_1) \exp( q - \m q) + (2 q_1 + \p q_1) \exp(\p q - q) \label{Toda-3},
		\end{align}
	\end{subequations}
	where either open-ended or periodic boundary conditions can be used. In Section \ref{sec-toda-ham} we will sketch a systematic construction of this hierarchy. A continuous Lagrangian 1-form for this hierarchy is known, where the configuration is represented by a vector in $\R^k$ containing the positions of all particles. This ignores the physical intuition behind the system, where we think of the particles on a discrete lattice in space. To capture this, along with the continuous time evolution, we develop a semi-discrete 2-form for the Toda lattice. 
	
	Before we get started, let us think about whether lattice shifts could be eliminated form the system \eqref{Toda-1}--\eqref{Toda-2}. By considering these equations as a linear system for the two exponential terms, we find an equivalent system
	\begin{subequations}
		\label{a-PDEs}
		\begin{align}
		\exp( \bar q - q) = \frac{1}{2}(q_2 + q_{11} - q_1^2) \label{a-PDE}, \\
		\exp( q - \underline q) = \frac{1}{2}(q_2 - q_{11} - q_1^2) \label{ma-PDE}.
		\end{align}
	\end{subequations}%
	This shows that we can eliminate one of the lattice shifts from Equations \eqref{Toda}. In fact, Equations \eqref{a-PDEs} can be understood as an NLS-type system
	\begin{align*}
	U_2 = U_{11} + 2U^2 V, \qquad V_2 = - V_{11} -  2 U V^2,
	\end{align*}
	with variables $U = \exp(\p q)$ and $V = \exp(q)$ \cite{adler1997class}. This observation relates the Toda lattice to a system of integrable PDEs by promoting a lattice shift to an additional dependent variable. It is far from obvious if this additional variable can be eliminated to obtain one scalar PDE. As we will see below, the Lagrangian multiform will provide a solution to this problem.

	Our construction is inspired on the known continuous 1-form, which can be obtained for example from the discrete-time Toda lattice using a continuum limit \cite{vermeeren2019continuum}, or from the variational symmetries of the system \cite{petrera2017variational}. Alternatively, it could be obtained from the Hamiltonian formulation of the hierarchy (see e.g.\@ \cite{suris2003problem}) using the methods of \cite{vermeeren2021hamiltonian}. Here we adapt the latter approach to yield a semi-discrete 2-form.
	
	\subsection{Hamiltonian formulation}
	\label{sec-toda-ham}
	
	The geometric structure of the Toda lattice is usually presented in Flaschka variables \cite{flaschka1974toda}
	\[ a = \exp(\bar q - q), \qquad b = q_1, \]
	with the Poisson brackets
	\[ \{a,b\} = a, \qquad \{a, \bar b\} = -a . \]
	We consider Hamilton functions
	\[ H_i = \sum_{\alpha \in \Z} \t^\alpha h_i = \ldots + \underline{h_i} + h_i + \overline{h_i} + \ldots, \]
	where
	\noeqref{Toda-h1,Toda-h2,Toda-h3}%to avoid labels being supressed
	\begin{subequations}
		\label{Toda-hs}
		\begin{align}
		h_1 &= \frac{1}{2} b^2 + a , \label{Toda-h1}\\
		h_2 &= \frac{1}{3} b^3 + a (b + \bar b) , \label{Toda-h2}\\
		h_3 &= \frac{1}{4} b^4 + a (b^2 + \bar b^2) + a b \bar b + a \bar a + \frac{1}{2} a^2 , \ \ldots \label{Toda-h3}
		\end{align}
	\end{subequations}%
	Note that the subscripts on $h$ and $H$ are labels, not derivatives.
	These Hamilton functions can be obtained from the usual Lax formulation of the Toda lattice, $\der{L}{t_1} = [B,L]$, by taking $H_i = \frac{1}{i}\mathrm{tr}{L^i}$  (see \cite{flaschka1974toda}, also \cite[Chapter 3]{suris2003problem}). We choose to write them as a sum $H_i = \sum_{\alpha \in \Z} \t^\alpha h_i$ in such a way that $h_i$ does not contain any negative shifts of $a$ or $b$.
	
	The corresponding equations of motion are
	\noeqref{ab-1,ab-2}%to avoid labels being supressed
	\begin{subequations}
		\label{ab-equations}
		\begin{align}
		& a_i = \{H_i,a\} = \left( \der{H_i}{\bar b} - \der{H_i}{b} \right) a , 
		\label{ab-1}\\
		& b_i = \{H_i,b\} = \Delta \left( \der{H_i}{a} a \right) . 
		\label{ab-2}
		\end{align}
	\end{subequations}
	In the original coordinates, the equations of motion are
	\noeqref{q-1,q-2}%to avoid labels being supressed
	\begin{subequations}
		\label{q-equations}
		\begin{align}
		& q_i = Q_i := \der{H_i}{b} , \label{q-1}\\
		& q_{1i} = B_i := \Delta \left( \der{H_i}{a} a \right) , \label{q-2}
		\end{align}
	\end{subequations}
	where the subscripts on $Q$ and $B$ are labels, not derivatives.
	
	An elementary calculation shows that
	\begin{equation}
	\label{poisson}
	\{ H_i, H_j \} = \sum_{\alpha \in \Z} \t^\alpha ( B_i Q_j - Q_i B_j) . 
	\end{equation}
	It is well-known that the $H_i$ are in Poisson involution, so this sum must be zero. It follows that the summand can be written as a difference:
	\begin{equation}\label{F-def}
	B_i Q_j - Q_i B_j = \Delta F_{ij}
	\end{equation}
	for some $F_{ij}$.
	
	\subsection{Semi-discrete 2-form}
	
	Following the construction of Lagrangian 1-forms from Hamiltonians in involution \cite{vermeeren2021hamiltonian}, we find a continuous Lagrangian 1-form for the Toda lattice with coefficients
	\[ L_{j} = \sum_\alpha \t^\alpha( q_1 q_j - h_j ), \]
	where the first few $h_j$ are given in Equation \eqref{Toda-hs}.
	We now look for a semi-discrete 2-form that reduces to this 1-form by the method of Remark \ref{remark-reduction}. This motivates the choice
	\begin{equation}
	\label{toda-L0j}
	 L_{0j} = q_1 q_j - h_j ,
	\end{equation}
	but does not guide our choice of coefficients $L_{ij}$.
	
	We want to construct coefficients $L_{ij}$ such that the exterior derivative of the semi-discrete 2-form 
	\[ \left( \sum_j L_{0j}\,\d t_j \ ,\ \sum_{i<j} L_{ij} \,\d t_i \wedge \d t_j \right) \]
	vanishes on solutions. Furthermore, in light of Proposition \ref{prop-mEL-P}, we would like it to attain a double zero on solutions.
	 The following fact, which can be thought of as a local version of Equation \eqref{poisson}, will come in useful.
	
	\begin{lemma}\label{lemma-dh}
		On the equations of motion \eqref{q-equations}, there holds
		\begin{align*}
		\D_i h_j = - B_j q_i + Q_j q_{1i} 
		+ \Delta \left( \der{H_j}{a} a \t q_i + \sum_{\alpha \geq 1} \sum_{\beta = 1}^{\alpha} \t^\beta \left( \der{(\t^{-\alpha} h_j)}{a} a_i + \der{(\t^{-\alpha} h_j)}{b} b_i \right) \right) .
		\end{align*}
	\end{lemma}
	\begin{proof}
		Since $h_j$ does not contain any negative shifts, we can write
		\begin{align*}
		\D_i h_j &= \sum_{\alpha \geq 0} \left( \der{h_j}{\t^\alpha a} \t^\alpha a_i + \der{h_j}{\t^\alpha b} \t^\alpha b_i \right) \\
		&= \sum_{\alpha \geq 0} \t^\alpha \left( \der{\t^{-\alpha} h_j}{a} a_i + \der{\t^{-\alpha} h_j}{b} b_i \right) \\
		&= \sum_{\alpha \geq 0}\left( \der{\t^{-\alpha} h_j}{a} a_i + \der{\t^{-\alpha} h_j}{b} b_i \right) 
		+ \sum_{\alpha \geq 1} \sum_{\beta=1}^\alpha \Delta\left( \t^\beta \left( \der{\t^{-\alpha} h_j}{a} a_i + \der{\t^{-\alpha} h_j}{b} b_i \right) \right) .
		\end{align*}
		Using once more that $h_j$ does not contain any negative shifts, it follows that
		\[ \D_i h_j = \der{H_j}{a} a (\t q_i - q_i) + \der{H_j}{b} q_{1i}
		+ \sum_{\alpha \geq 1} \sum_{\beta=1}^\alpha \Delta\left( \t^\beta \left( \der{\t^{-\alpha} h_j}{a} a_i + \der{\t^{-\alpha} h_j}{b} b_i \right) \right) .
		\]
		To finish the proof, observe that the first term in the right hand side is equal to
		\[ \der{H_j}{a} a (\t q_i - q_i) = -\Delta\left(\der{H_j}{a} a\right) q_i + \Delta\left( \der{H_j}{a} a \t q_i \right) \]
		and use Equations \eqref{q-equations}.
	\end{proof}
	
	The computation in the proof of Theorem \ref{thm-Lij} below, which uses Lemma \ref{lemma-dh}, shows that 
	the exterior derivative $\dd \cL$ attains a double zero on solutions to the Toda hierarchy if we set
	\begin{equation}
	\label{toda-Lij}
	\begin{split}
	L_{ij} &= \der{H_i}{a} a \t q_j + \sum_{\alpha \geq 1} \sum_{\beta = 1}^{\alpha} \t^\beta \left( \der{(\t^{-\alpha} h_i)}{a} a_j + \der{(\t^{-\alpha} h_i)}{b} b_j \right) \\
	&\quad - \der{H_j}{a} a \t q_i - \sum_{\alpha \geq 1} \sum_{\beta = 1}^{\alpha} \t^\beta \left( \der{(\t^{-\alpha} h_j)}{a} a_i + \der{(\t^{-\alpha} h_j)}{b} b_i \right) - F_{ij},
	\end{split}
	\end{equation}
	where $F_{ij}$ is as in Equation \eqref{F-def}.
	
	\begin{thm}\label{thm-Lij}
		Let $\cL$ be the semi-discrete 2-form with coefficients given by Equations \eqref{toda-L0j} and \eqref{toda-Lij}. There holds $\delta \dd \cL = 0$ on solutions to the Toda hierarchy, hence the Toda hierarchy implies the multi-time Euler-Lagrange equations.
	\end{thm}
	\begin{proof}
		Equation \eqref{exterior-der} states that $\dd \cL$ has coefficients $P_{ijk} = \D_i L_{jk} - \D_j L_{ik} + \D_k L_{ij}$ and
		\begin{align*}
		P_{0ij} & = \Delta L_{ij} - \D_i L_{0j} + \D_j L_{0i} \\
		&=  -q_{1i} q_j + q_{1j} q_i + \D_i h_j - \D_j h_i + \Delta L_{ij} .
		\end{align*} 
		Using Lemma \ref{lemma-dh} we find
		\begin{align}
		P_{0ij} 
		&= -q_{1i} q_j + q_{1j} q_i - B_j q_i + Q_j q_{1i} + B_i q_j - Q_i q_{1j} - \Delta F_{ij} \notag \\
		&= -(q_{1i} - B_i) (q_j - Q_j) + (q_{1j} - B_j) (q_i - Q_i). \label{Toda-closure}
		\end{align}
		Hence $P_{0ij}$ attains a double zero on solutions.
		
		In addition, we have
		\begin{align*}
		\Delta P_{ijk} &= \D_i \Delta L_{jk} - \D_j \Delta L_{ik} + \D_k \Delta L_{ij} \\
		&= \D_i P_{0jk} - \D_j P_{0ik} + \D_k P_{0ij},
		\end{align*}
		which also attains a double zero on solutions. Therefore, on solutions,
		\begin{align*}
		\der{\Delta P_{ijk}}{q_I^{[n]}}=   \der{ P_{ijk}}{q_I^{[n]}}-\der{\t^{-1} P_{ijk}}{q_I^{[n]}}=0
		\end{align*}
		for all $i,j,k,I$, and $n$.  Using Lemma \ref{lemma-commute-ders}, this becomes
		\begin{align}\label{eqn-redf}
		\der{ P_{ijk}}{q_I^{[n]}}-\t^{-1}\der{ P_{ijk}}{q_I^{[n+1]}}=0.    
		\end{align}
		For every $I$ such that, for some $n$, $q^{[n]}_I$ appears in $P_{ijk}$, we let $n_{max}$ be the maximum $n$ such that $q^{[n]}_I$ appears in $P_{ijk}$.  Then \eqref{eqn-redf} tells us that, on solutions, 
		\begin{align*}
		\der{ P_{ijk}}{q_I^{[n_{max}]}}=0. 
		\end{align*}
		It then follows inductively from \eqref{eqn-redf} that 
		\begin{align*}
		\der{ P_{ijk}}{q_I^{[n]}}=0 
		\end{align*}
		for all $I$ and $n$, or equivalently that $\delta P_{ijk} = 0$ on solutions. Hence, on solutions to the Toda hierarchy, $\delta \dd \cL = 0$. 
		Finally, by Proposition \ref{prop-mEL-P} this means that the multi-time Euler-Lagrange equations are consequences of the Toda hierarchy.
	\end{proof}
	
	\subsection{Explicit calculations}
	
	Using the Hamiltonians $h_1,h_2$ from Equation \eqref{Toda-hs}, we find
	\begin{align*}
	&L_{01} = \frac{1}{2} q_1^2 - \exp(\p q - q) \\
	&L_{02} = q_1 q_2 - \frac{1}{3} q_1^3 - (q_1 + \p q_1) \exp(\p q - q) 
	\end{align*}
	Theorem \ref{thm-Lij} gives us the coefficient
	\[ L_{12} = -(b + \p b) a \p q_1 - a \p b_1 + a \p q_2 - F_{12}, \]
	where $F_{12}$ should satisfy
	\begin{align*}
	\Delta F_{12}
	&= B_1 Q_2 - Q_1 B_2 \\
	&= (\Delta \m a)(b^2 + a + \m a) - b \Delta( (b + \m b) \m a)  \\
	&= a^2 - \m a^2 - b \bar{b} a + b \m b \m a \\
	&= \Delta( a^2 - \p b b a).
	\end{align*}
	Hence for $L_{12}$ we could take
	\begin{align*}
	L_{12} &= a \p q_2 - (q_1 + \p q_1) a \p q_1 - a \p q_{11} - a^2 + \p q_1 q_1 a \\
	&= - a ( \p q_1^2 + \p q_{11} - \p q_2 + a) \\
	&= - \left( \frac{1}{2}(\p q_1^2 + \p q_{11} - \p q_2) + a \right)^2 + \frac{1}{4}(\p q_1^2 + \p q_{11} - \p q_2)^2.
	\end{align*}
	From Equation \eqref{ma-PDE} we see that the first term attains a double zero on solutions, hence we obtain an equivalent semi-discrete two-form if we leave it out and take
	\begin{equation*}%\label{toda-L12}
	L'_{12} = \frac{1}{4}(\p q_1^2 + \p q_{11} - \p q_2)^2 .
	\end{equation*}
	Note that the factorisation \eqref{Toda-closure} is valid for $L_{12}$. With $L'_{12}$ we would get a different expression for $P_{012}$, which also attains a double zero on solutions.
	
	Theorem \ref{thm-Lij} now implies that $q$ satisfies the Toda equations \eqref{Toda} if and only if it is critical for the semi-discrete 2-form $\cL = \left( L_{01} \,\d t_1 + L_{02} \,\d t_2 \,,\,  L'_{12} \,\d t_1 \wedge \d t_2 \right)$. Indeed we can recover the first Toda equation from the variational principle by integrating $\cL$ over the semi-discrete surface spanned by $t_1$ and the discrete direction (i.e.\@ consisting of copies at each lattice site of a line in the $t_1$-direction). Similarly, the second Toda equation can be obtained using the semi-discrete surface spanned by $t_2$ and the discrete direction. There are many other semi-discrete surfaces we could consider. Of particular interest are those that consist only of the $(t_1,t_2)$-plane at one single lattice site. The resulting Euler-Lagrange equation is the subject of the following corollary.
	
	\begin{cor}
		The Toda hierarchy \eqref{Toda} implies the PDE
		\begin{equation}
		\label{pde-q22}
		\frac{1}{2} q_{22} - q_{11} q_2 - 2 q_{12} q_1 - \frac{1}{2} q_{1111} + 3 q_1^2 q_{11} = 0 .
		\end{equation}
	\end{cor}
	\begin{proof}
		The PDE is obtained as the shifted multi-time Euler-Lagrange equation 
		\[ \t^{-1} \var{12}{L_{12}'}{\p q} = 0 .\]
		Hence, by virtue of Theorem \ref{thm-Lij}, it is implied by the Toda hierarchy.
	\end{proof}
	
	This result indicates that PDE \eqref{pde-q22} is integrable in its own right, since Equation \eqref{a-PDEs} provides an auto-Bäcklund transformation for it. Indeed, Equation \eqref{pde-q22} can be identified as an integrable Boussinesq-type equation. In particular, it is equivalent to Equation (66) of \cite{mikhailov2007classification} via $u = q_1$ and $v = q_2$, and to Equation (1.2) of \cite{kupershmidt1985mathematics} via $u = 2 q_1$ and $h = 2 q_2 -  2 q_1^2$. 
	
	While it is not entirely surprising that a higher-order PDE can be obtained by eliminating lattice shifts form Equations \eqref{Toda}, doing this by direct computation would be tedious. It is remarkable that from our semi-discrete Lagrangian 2-form it follows immediately. 
	This sheds a new light on the observation that integrable PDEs are connected to differential-difference equations \cite{levi1980backlund,levi1981nonlinear}. 
	It is yet another indication that Lagrangian multiform theory captures integrable hierarchies in a fundamental way.
	
	Using the same methods as above we obtain
	\begin{align*}
	& L_{03} = q_1 q_3 - \frac{1}{4} q_1^4 - a (q_1^2 + \p q_1^2 + q_1 \p q_1) - a \p a + \frac{1}{2} a^2, \\
	& L_{13} = - a \left( \p q_1^3 + 2 a \p q_1 + \p a \p{\p q}_1 + 2 \p q_1 \p q_{11} + q_1 \p q_{11} - \p q_3 + a q_1 - \p a q_1 \right),
	\end{align*}
	and
	\begin{align*}
	L_{23} &= - a \Big( \p q_2 \left( q_1^2 + \p q_1^2 + q_1 \p q_1 + \m a + a \right) + \p{\p q}_2 \p a + 2 \p q_1 \p q_{12} + q_1 \p q_{12} 
	- \p q_3 \left( q_1 + \p q_1 \right) - \p q_{13} \\
	&\qquad\qquad - q_1^2 \p q_1^2 - \m a \p q_1^2 + 2 a q_1 \p q_1 - \p a q_1^2 - a \p a - \m a \p a - \m a a - a^2 \Big) .
	\end{align*}
	Again we can use the multi-time Euler-Lagrange equations to obtain a PDE at a single lattice site:
	\begin{cor}
		The Toda hierarchy \eqref{Toda} implies the PDE
		\begin{equation}
		\label{pde-q3}
		q_3 = -2 q_1^3 + 3 q_1 q_2 + q_{111} .
		\end{equation}
	\end{cor}
	\begin{proof}
		From the Euler-Lagrange equation $\displaystyle \var{13}{L_{13}}{\bar q} = 0$ we obtain
		\[ q_1^3 - 3 q_1 q_{11} + 6 q_1 a + q_{111} - q_3 = 0, \]
		which we can write using Equation \eqref{a-PDE} as
		\[ -2 q_1^3 + 3 q_1 q_2 + q_{111} - q_3 = 0 . \qedhere \]
	\end{proof}
	
	At this stage, it is unclear whether Equation \eqref{pde-q3} by itself is integrable. However, the system of Equations \eqref{pde-q22}--\eqref{pde-q3} is integrable in the sense of existence of an auto-Bäcklund transformation, given by Equation \eqref{a-PDEs}: if $q$ solves both PDEs , then so does $\bar{q} = q + \log\left( \frac{1}{2}(q_2 + q_{11} - q_1^2) \right)$, as can be verified by a long but elementary calculation. A detailed investigation of Equation \eqref{pde-q3}, as part of the hierarchy of higher equations which can presumably be obtained in an analogous way, is left for future work.
	
	\section{Semi-discrete potential KdV}
	\label{sec:sdpkdv}
	
	As another example of a system of interacting particles on a line, we consider the semi-discrete potential KdV equation. It belongs to the class of equations studied by direct linearisation in \cite{nijhoff1983direct} and appears as a semi-continuous limit of the lattice potential KdV equation \cite{wiersma1987lattice}. The semi-discrete potential KdV equation and the second member of its hierarchy can be written as
	\begin{align}
	&q_1 = \frac{\alpha}{\alpha + \p q - \m q} - \beta 
	\label{sdpKdV-1} , \\
	&q_2 = \frac{-\alpha^2}{(\alpha + \p q - \m q)^2} \left( \frac{1}{\alpha + \p{\p q} - q} + \frac{1}{\alpha + q - \m{\m q}} \right) ,
	\label{sdpKdV-2}
	\end{align}
	for constants $\alpha$ and $\beta$. In \cite{wiersma1987lattice} we find this hierarchy with $\beta = 1$, and with a different second equation which is a linear combination of our Equations \eqref{sdpKdV-1} and \eqref{sdpKdV-2}. Solutions of equations \eqref{sdpKdV-1} and \eqref{sdpKdV-2} give critical points of the actions associated to the Lagrangians
	\begin{align*}
	& L_{01} = q_1 \p q - \alpha \log(\alpha + \p q - \m q), \\
	& L_{02} = q_2 \p q - \frac{\alpha^2}{(\alpha + \p q - \m q)(\alpha + \p{\p q} - q)} .
	\end{align*}
	Notice that $L_{01}$ does not depend on $\beta$. Indeed, its Euler-Lagrange equation is
	\[ \p q_1 - \m q_1  = \frac{\alpha}{\alpha + \p{\p q} - q} - \frac{\alpha}{\alpha + q - \m{\m q}}, \]
	which is implied by Equation \eqref{sdpKdV-1} (regardless of the value of $\beta$) but not equivalent to it.
	
	To find a semi-discrete Lagrangian two-form we calculate $\D_1 L_{02} - \D_2 L_{01}$ and write it as a discrete derivative. To keep the length of our expressions under control we will write
	\[ v = \frac{1}{\alpha + \p q - \m q}. \]
	An elementary calculation shows that with 
	\[ L_{12} = \alpha^2 \p v^2 \left( v \p{\p q}_1 + \p{\p v} q_1 \right) + \alpha^2 \beta \p v^2 \left( \p{\p v} + v \right) - \alpha^3 v \p v^2 \p{\p v} + \alpha \p v q_2 + \beta \p q_2 \]
	there holds
	\begin{align*}
	\Delta L_{12} - \D_1 L_{02} + \D_2 L_{01} &= 
	\left( \p q_2 + \alpha^2 \p v^2 \left( \p{\p v} + v \right) \right) \left( q_1 - \alpha v + \beta \right) \\
	&\qquad - \left( q_2 + \alpha^2 v^2 \left( \p v + \m v \right) \right) \left( \p q_1 - \alpha \p v + \beta \right),
	\end{align*}
	so the exterior derivative of the semi-discrete 2-form $\cL = \left( L_{01} \,\d t_1 + L_{02} \,\d t_2 ,  L_{12} \,\d t_1 \wedge \d t_2 \right)$ attains a double zero on solutions to the semi-discrete potential KdV hierarchy. We can check that its Euler-Lagrange equations are equivalent to this hierarchy. For example, we have
	\[ \var{12}{L_{12}}{\p q_2} = -\var{01}{L_{01}}{\p q} \quad\Rightarrow\quad q_1 = \alpha v - \beta \]
	and
	\[ \var{12}{L_{12}}{\p q_1} = \var{02}{L_{02}}{\p q} \quad\Rightarrow\quad q_2 = -\alpha^2 v^2 \left( \p v + \m v \right) .\]
	Note that these equations are stronger than the Euler-Lagrange equations of $L_{01}$ and $L_{02}$ individually.
	
	\begin{remark}
		The semi-discrete potential KdV hierarchy is closely related to the Volterra hierarchy. Its leading equation is given by 
		\begin{equation}\label{volterra}
		a_1 = a (\p a - \m a)
		\end{equation} 
		and can be obtained from the semi-discrete potential KdV hierarchy by defining 
		\[ a = \alpha v \m v ,\]
		see \cite[Exercise 5.4]{hietarinta2016discrete}. The Volterra lattice \eqref{volterra} can be viewed as a generalisation of the Lotka-Volterra predator-prey model to a chain of $n$ species, each of which is preyed upon by the next. As an integrable system it first appeared in \cite{manakov1974complete} and \cite{kac1975explicitly}. It is part of a hierarchy that can be obtained by restricting the even-numbered flows of the Toda hierarchy to the manifold defined by $b = 0$. Alternatively, it can be obtained from the Toda lattice by a Miura transformation \cite[Chapter 4]{suris2003problem}. To our knowledge, no direct Lagrangian description of the Volterra hierarchy is known.
	\end{remark}

	\section{Conclusions}
	
	We have presented the semi-discrete theory of Lagrangian multiforms, with the Toda lattice as our leading example. While the main text considers only the case of a single discrete independent variable, the general theory is analogous and outlined in the appendix. 
	
	The ideas of this paper closely follow the multiform theory in the fully discrete and continuous cases. Nevertheless, it led to an unexpected result: the semi-discrete multiform formulation of the Toda hierarchy produces PDEs which hold at a single lattice site as a consequence of the differential-difference equations of the Toda hierarchy. This phenomenon showcases the utility of the Lagrangian multiform approach in this context, and is a strong motivation to develop the multiform formulation of other semi-discrete hierarchies.
	
	\section*{Acknowledgements}
	
	MV acknowledges support by the Deutsche Forschungsgemeinschaft (DFG), project number VE~1211/1-1.
	
	We would like to thank Vladimir Novikov for his help to identify Equation \eqref{pde-q22} as a Boussinesq-type equation, and Vincent Caudrelier and Frank Nijhoff for inspiring discussions on many topics related to this work.

	\section*{Appendix: general semi-discrete multi-time EL equations}
	
	We consider a semi-discrete space $\Z^M \times \R^{N-M}$ of independent coordinates $n_1,\ldots,n_M$, $t_{M+1},\ldots,t_N$,  and dependent variable $q$.  We define the shift operator $\t_i$ such that
	\begin{equation}
	\t_i q(n_1,\ldots,n_i,\ldots,n_M,t_{M+1},\ldots,t_N)=q(n_1,\ldots,n_i+1,\ldots,n_M,t_{M+1},\ldots,t_N).
	\end{equation}
	We will use the notation $\D_i$ for difference operator or the total derivative, depending on whether $i$ represents a discrete or continuous direction.
	For $1\leq i\leq M$ we define
	\begin{equation}
	\D_i=\t_i-\textsf{id} ,
	\end{equation}
	where $\textsf{id}$ represents the identity operator. This differs from the definition of the discrete derivative $\Delta=\textsf{id}-\t^{-1}$ given in Section \ref{sec:geometry}.  The reasons for this difference are discussed at the end of this appendix. For $M<i\leq N$ we define
	\begin{equation}
	\D_i=\frac{\partial}{\partial t_i} +\sum_Iq_{It_i}\frac{\partial}{\partial q_I} ,
	\end{equation}
	where $I$ is an $N$-component multi-index $(i_1,\ldots,i_N)$ representing shifts with respect to $n_1,\ldots,n_M$ and derivatives with respect to $t_{M+1},\ldots,t_N$.  We shall also use the notation $I \alpha$ to mean $(i_1,\ldots,i_\alpha+1,\ldots,i_N)$ and $I\setminus \alpha$ to mean $(i_1,\ldots,i_\alpha-1,\ldots,i_N)$.
	
	We introduce symbols $\textsf{d}n_i$, which are the discrete analogue of the $\textsf{d}t_i$. They are the same as the $\Delta_i$ from \cite{Mansfield2008}. In the exterior algebra spanned by the $\textsf{d}n_i$ and $\textsf{d}t_i$ we consider a $k$-form
	\begin{equation}\label{kformLk2d}
	\cL=\sum_{1\leq i_1<\ldots <i_k\leq N}L_{(i_1\ldots i_k)}\ \textsf{d}n_{i_1}\wedge\ldots\wedge \textsf{d}n_{i_j}\wedge\textsf{d}t_{i_{j+1}}\wedge\ldots\wedge\textsf{d}t_{i_k} ,
	\end{equation}
	where $j$ is the largest integer such that $i_j \leq M$. We assume that each $L_{(i_1\ldots i_k)}$ depends on $q$ and shifts of $q$ in the $n_1,\ldots,n_M$ coordinates (without loss of generality, we shall assume that there are no backward shifts), derivatives of $q$ in the $t_{M+1},\ldots,t_N$ coordinates and combinations thereof.  
	Even though $\cL$ is formally a $k$-form, only the $\d t_i$ have an interpretation as differentials. The $\d n_i$ are formal symbols, so the geometric interpretation of $\textsf{d}n_{i_1}\wedge\ldots\wedge \textsf{d}n_{i_j}\wedge\textsf{d}t_{i_{j+1}}\wedge\ldots\wedge\textsf{d}t_{i_k}$ is a $(k-j)$-form. Hence, geometrically, Equation \eqref{kformLk2d} is a differential form of mixed type (a special case of which was considered in Definition \ref{def-form}), but computationally it is treated in very close analogy to a proper differential $k$-form. Using the definitions of \cite{Mansfield2008} we find
	\begin{equation}
	\dd\cL=\sum_{1\leq i_1<\ldots<i_{k+1}\leq N}A^{i_1\ldots i_{k+1}} \,\textsf{d}n_{i_1}\wedge\ldots\wedge\textsf{d}n_{i_j}\wedge \textsf{d}t_{i_{j+1}}\wedge\ldots\wedge\textsf{d}t_{i_{k+1}} ,
	\end{equation}
	where the $A^{i_1\ldots i_{k+1}}$ are given by
	\begin{equation}\label{Adef}
	A^{i_1\ldots i_{k+1}}=\sum_{\alpha=1}^{k+1}(-1)^{(\alpha+1)}\D_{{i_\alpha}}L_{(i_1,\ldots,i_{\alpha-1},i_{\alpha+1},\ldots,i_{k+1})}.
	\end{equation}
	
	For a fixed $i_1,\ldots ,i_{k+1}$, we shall write $L_{(\bar \alpha)}$ to denote $L_{(i_1,\ldots,i_{\alpha-1},i_{\alpha+1},\ldots,i_{k+1})}$.  A multi-index denoted by $J$ is such that component $j_i=0$ whenever $i\neq i_1,\ldots, i_{k+1}$, i.e.\@ $J$ represents shifts with respect to $n_{i_1}, \ldots,n_{i_j},t_{i_{j+1}},\ldots,t_{i_{k+1}}$.  We define the variational derivative with respect to $q_I$ acting on $L_{(\bar \alpha)}$ and $A^{i_1\ldots i_{k+1}}$ respectively as
	\begin{equation}\label{vardivdef}
	\begin{split}
	&\frac{\delta L_{(\bar \alpha)}}{\delta q_{I}}=\sum_{\substack{  J\\j_{i_\alpha}=0}}(\t^{-1})_{J^n}(-\D)_{J^t}\frac{\partial L_{(\bar \alpha)}}{\partial q_{IJ}} ,\\
	&\frac{\delta A^{i_1\ldots i_{k+1}}}{\delta q_{I}}=\sum_J(\t^{-1})_{J^n}(-\D)_{J^t}\frac{\partial A^{i_1\ldots i_{k+1}}}{\partial q_{IJ}} ,
	\end{split}
	\end{equation}
	where $I$ is again an $N$ component multi-index $(i_1,\ldots,i_N)$ representing shifts with respect to $n_1,\ldots,n_M$ and derivatives with respect to $t_{M+1},\ldots,t_N$.  We use the notation $I^n$ to denote only the first $M$ components of $I$ that relate to shifts in the $n_i$ coordinates, and $I^t$ to denote the last $N-M$ components of $I$ that relate to derivatives with respect to the $t_i$.  Therefore,
	\begin{equation}
	q_I=\t_{I^n}\D_{I^t}q=\t_1^{i_1}\ldots\t_N^{i_M}\D_{M+1}^{i_{M+1}}\ldots\D_N^{i_N}q.
	\end{equation}
	We define that a variational derivative with respect to $q_I$ is zero in the case where any component of the multi-index $I$ is negative (we are only able to do this because we have assumed that there are no backward shifts in our Lagrangians).  We note that in contrast to the variational derivative operators defined in Section \ref{sec:theory}, for brevity of notation, we now omit an index on the operator.  For example, in this appendix we write $\var{}{L_{ij}}{q_{I}}$ instead of $\var{ij}{L_{ij}}{q_{I}}$.
	
	\begin{thm}[Multi-time Euler-Lagrange equations]\label{mfeqnsthm}
		The function $q$ is a critical point of the $k$-form $\cL$ as defined in \eqref{kformLk2d} if and only if for all $i_1,\ldots i_{k+1}$ such that $1\leq i_1<\ldots<i_{k+1}\leq N$, and for all $I$,
		\begin{equation}\label{MFELeqnsA}
		\frac{\delta}{\delta q_I}A^{i_1\ldots i_{k+1}}=0,
		\end{equation}
		or equivalently,
		\begin{equation}\label{MFELeqns}
		\sum_{\alpha=1}^{j}(-1)^{\alpha + 1}\t_\alpha\frac{\delta L_{(\bar \alpha)}}{\delta q_{I\setminus i_\alpha}}+\sum_{\alpha=j+1}^{k+1}(-1)^{\alpha +1}\frac{\delta L_{(\bar \alpha)}}{\delta q_{I\setminus i_\alpha}}=0,
		\end{equation}
		where $j$ is the largest integer such that $i_j \leq M$.
	\end{thm}

	\noindent In order to prove that these are the multi-time EL equations, we will require the following lemma:
	
	\begin{lemma}\label{discpartlem}
		Let $1\leq i_1<\ldots<i_{k+1}\leq N$ be fixed.  For all multi-indices $I$,
		\begin{equation}\label{lem1}
		\frac{\partial A^{i_1\ldots i_{k+1}}}{\partial q_I}=\sum_{\substack{J\\j_i\leq 1}}(-\t^{-1})_{J^n}\D_{J^t}\frac{\delta A^{i_1\ldots i_{k+1}}}{\delta q_{IJ}} ,
		\end{equation}
		where the summation is over all multi-indices $J$ such that  $j_i=0$ whenever $i\neq i_1,\ldots, i_{k+1}$ and the non-zero $j_i$ are equal to 1.
	\end{lemma}
	
	\begin{proof}
		We first notice that the partial derivative on the left hand side of \eqref{lem1} appears only once in the sum on the right hand side.  We now need to show that all other terms that appear on the right hand side of \eqref{lem1}, sum to zero.  We note that all terms on the right hand side of \eqref{lem1} are of the form 
		\begin{equation}
		\label{Kterm}
		(\t^{-1})_{K^n}\D_{K^t}\frac{\partial A^{i_1\ldots i_{k+1}}}{\partial q_{IK}}
		\end{equation}
		for some multi-index $K = (k_1,\ldots,k_N)$ which satisfies $k_i=0$ whenever $i\neq i_1,\ldots, i_{k+1}$.
		Let $r$ be the number of non-zero entries in $K$.  We notice that the term \eqref{Kterm} appears exactly once when $|J|=0$, exactly ${r\choose 1}$ times with a factor of $-1$ when $|J|=1$, exactly ${r\choose 2}$ times when $|J|=2$ etc... In total, this term appears with a factor of $\sum_{i=0}^{r}(-1)^i{r\choose i}$.  It can easily be seen that this sum is zero by considering the binomial expansion of $(1-1)^r$.
	\end{proof}
	
	\begin{proof}[Proof of Theorem \ref{mfeqnsthm}]
		The first part of the proof of Proposition \ref{prop-mEL-P} of Section \ref{sec:theory}, showing that criticality is equivalent to $\delta\dd \cL=0$, immediately extends to the present case. We note that the equations given by $\delta\dd \cL=0$ are equivalent to
		\begin{equation}
		\frac{\partial}{\partial q_I}A^{i_1\ldots i_{k+1}}=0
		\end{equation}
		for all $I$ and all $1\leq i_1<\ldots< i_{k+1}\leq N$. Using Lemma \ref{discpartlem}, we see that this is the case if and only if Equation \eqref{MFELeqnsA} holds for all $I$ and all $1\leq i_1<\ldots< i_{k+1}\leq N$. It remains to show that \eqref{MFELeqns} and \eqref{MFELeqnsA} are equivalent expressions.
		
		The identities
		\begin{equation}
		\frac{\partial}{\partial q_I}\t_{i}=\t_{i}\frac{\partial}{\partial q_{I\setminus i}}
		\end{equation}
		for $1\leq i\leq M$ and
		\begin{equation}
		\frac{\partial}{\partial q_I}\D_{i}=\frac{\partial}{\partial q_{I\setminus i}}+\D_{i}\frac{\partial}{\partial q_I}
		\end{equation}
		for $M<i\leq N$ tell us that
		\begin{equation}
		\frac{\partial}{\partial q_I}A^{i_1\ldots i_{k+1}}=\sum_{\alpha=1}^{j}(-1)^{\alpha+1}\bigg(  \t_{{i_\alpha}}\frac{\partial L_{(\bar\alpha)}}{\partial q_{I\setminus i_\alpha}}     -\frac{\partial L_{(\bar\alpha)}}{\partial q_{I}}\bigg)+\sum_{\alpha=j+1}^{k+1}(-1)^{\alpha+1}\bigg(\frac{\partial L_{(\bar\alpha)}}{\partial q_{I\setminus i_\alpha}}+\D_{{i_\alpha}}\frac{\partial L_{(\bar\alpha)}}{\partial q_{I}}\bigg)
		\end{equation}
		so
		\begin{equation}\label{1stdeltaPd}
		\begin{split}
		\frac{\delta}{\delta q_I}A^{i_1\ldots i_{k+1}}&=\sum_J(\t^{-1})_{J^n}(-\D)_{J^t}\frac{\partial}{\partial q_{IJ}}A^{i_1\ldots i_{k+1}}\\
		&=\sum_J(T^{-1})_{J^n}(-\D)_{J^t}\sum_{\alpha=1}^{j}(-1)^{\alpha+1}\bigg(\t_{{i_\alpha}}\frac{\partial L_{(\bar\alpha)}}{\partial q_{IJ\setminus i_\alpha}}     -\frac{\partial L_{(\bar\alpha)}}{\partial q_{IJ}}\bigg)\\
		&\qquad +\sum_J(\t^{-1})_{J^n}(-\D)_{J^t}\sum_{\alpha=j+1}^{k+1}(-1)^{\alpha+1}\bigg(\frac{\partial L_{(\bar\alpha)}}{\partial q_{IJ\setminus i_\alpha}}+\D_{{i_\alpha}}\frac{\partial L_{(\bar\alpha)}}{\partial q_{IJ}}\bigg).
		\end{split}
		\end{equation}
		For $1\leq \alpha\leq j$, whenever $j_{i_\alpha}\neq 0$ in this sum, so $J$ is of the form $Ki_\alpha$ for some multi-index $K$, then
		\begin{equation}
		\pm(\t^{-1})_{J^n}(-\D)_{J^t}\t_{{i_\alpha}}\frac{\partial L_{(\bar\alpha)}}{\partial q_{IJ\setminus i_\alpha}}=\pm(\t^{-1})_{K^n}(-\D)_{K^t}\frac{\partial L_{(\bar\alpha)}}{\partial q_{IK}}
		\end{equation}
		will appear in this sum.  When $J=K$, the term
		\begin{equation}
		\mp(\t^{-1})_{K^n}(-\D)_{K^t}\frac{\partial L_{(\bar\alpha)}}{\partial q_{IK}}
		\end{equation}
		will appear, so these two terms cancel.  Similarly, for $j+1\leq \alpha\leq k$, whenever $j_{i_\alpha}\neq 0$ in this sum, so $J$ is of the form $Ki_\alpha$ for some multi-index $K$, then
		\begin{equation}
		\pm(T^{-1})_{J^n}(-\D)_{J^t}\frac{\partial L_{(\bar\alpha)}}{\partial q_{IJ\setminus i_\alpha}}=\mp(T^{-1})_{K^n}\D_{{i_\alpha}}(-\D)_{K^t}\frac{\partial L_{(\bar\alpha)}}{\partial q_{IK}}
		\end{equation}
		will appear in this sum.  When $J=K$, the term
		
		\begin{equation}
		\pm(T^{-1})_{K^n}(-\D)_{K^t}\D_{{i_\alpha}}\frac{\partial L_{(\bar\alpha)}}{\partial q_{IK}}
		\end{equation}
		will appear, so these two terms cancel.  As a result, \eqref{1stdeltaPd} simplifies to
		\begin{equation}
		\begin{split}
		\frac{\delta}{\delta q_I}A^{i_1\ldots i_{k+1}}
		&=\sum_{\alpha=1}^{j}\,\sum_{\substack{ J\\j_{i_\alpha}=0}}(-1)^{\alpha+1}(\t^{-1})_{J^n}(-\D)_{J^t}\t_{i_\alpha}\frac{\partial L_{(\bar\alpha)}}{\partial q_{IJ\setminus i_\alpha}}\\
		&\qquad +\sum_{\alpha=j+1}^{k+1}\,\sum_{\substack{ J\\j_{i_\alpha}=0}}(-1)^{\alpha+1}(\t^{-1})_{J^n}(-\D)_{J^t}\frac{\partial L_{(\bar\alpha)}}{\partial q_{IJ\setminus i_\alpha}}\\
		&=\sum_{\alpha=1}^{j}(-1)^{\alpha+1}\t_{i_\alpha}\frac{\delta L_{(\bar \alpha)}}{\delta q_{I\setminus i_{\alpha}}}+\sum_{\alpha=j+1}^{k+1}(-1)^{\alpha+1}\frac{\delta L_{(\bar \alpha)}}{\delta q_{I\setminus i_{\alpha}}}.
		\end{split}
		\end{equation}
		This shows that $\delta\dd \cL=0$ is equivalent to
		\begin{equation*}
		\sum_{\alpha=1}^{j} (-1)^{\alpha+1}\t_{i_\alpha}\frac{\delta L_{(\bar \alpha)}}{\delta q_{I\setminus i_{\alpha}}} + \sum_{\alpha=j+1}^{k+1}(-1)^{\alpha+1}\frac{\delta L_{(\bar \alpha)}}{\delta q_{I\setminus i_{\alpha}}} = 0 . \qedhere
		\end{equation*}
	\end{proof}
	
	In this appendix, we defined the discrete derivarive $\D_i=\t_i-\textsf{id}$.  Alternatively, we could have defined $\D_i=\textsf{id}-\t_i^{-1}$ (as we did in Section \ref{sec:geometry}) which would have led to the equivalent multi-time EL equations
	\begin{equation}\label{altMFeqns}
	\sum_{\alpha=1}^{j}(-1)^{\alpha + 1}\frac{\delta \tilde L_{(\bar \alpha)}}{\delta q_{I}}+\sum_{\alpha=j+1}^{k+1}(-1)^{\alpha +1}\frac{\delta \tilde L_{(\bar \alpha)}}{\delta q_{I\setminus i_\alpha}}=0.
	\end{equation}
	We use $\tilde L$ to denote the Lagrangians because they are not the same as the ones in \eqref{MFELeqns}. They are related by
	\begin{equation}
	\tilde L_{\bar \alpha}=\prod_{\substack{\beta=1\\ \beta \neq \alpha}}^j \t_{i_\beta}^{-1} L_{\bar \alpha},
	\end{equation}
	i.e.\@ by a shift in all discrete directions except the direction labelled by $i_\alpha$ (if it is discrete).
	The equivalence of \eqref{MFELeqns} and \eqref{altMFeqns} can then be seen by applying $\prod_{\beta=1}^j \t_{i_\beta}^{-1}$ to \eqref{MFELeqns} and re-labeling the multi-index $I$ to obtain \eqref{altMFeqns}.
	
	We choose to present the general semi-discrete multi-time EL equations in the form given in \eqref{MFELeqns} in order to highlight the close connection between the semi-discrete and continuous multi-time EL equations, with $I \setminus i_\alpha$ appearing in both.  Also, when the equations are presented in this way, it is clear that they include the usual EL equations for each $L_{\bar \alpha}$ (obtained by setting $I=i_\alpha$). On the other hand, the multi-time EL equations that we presented in Section \ref{sec:theory} are in the form given in \eqref{altMFeqns} in order to avoid the presence of shift operators in the multi-time EL equations.
	
	\bibliographystyle{abbrvnat_mv}
	\bibliography{semidiscrete}
	
\end{document}